\documentclass[12pt]{amsart}
\usepackage[leqno]{amsmath}
\usepackage{amssymb}
\usepackage{amsthm}
\usepackage[mathscr]{eucal}
\usepackage{amscd}
\usepackage{url}
\usepackage{color}
\usepackage{hyperref}
\newtheorem{thm}{Theorem}
\newtheorem{lem}[thm]{Lemma}
\newtheorem{crl}[thm]{Corollary}
\newtheorem{rmk}{Remark}
\newtheorem{nt}{Note}
\newtheorem{Example}{Example}
\numberwithin{equation}{section}
\numberwithin{thm}{section}
\numberwithin{rmk}{section}
\numberwithin{nt}{section}
\numberwithin{Example}{section}

\makeatletter
\let\oldcite\cite
\def\cite#1{\expandafter\def\csname#1\endcsname{}\oldcite{#1}}
\let\oldbibitem\bibitem
\def\bibitem#1{\@ifundefined{#1}{\typeout{#1 nonused}}
{\typeout{#1 used}}\oldbibitem{#1}}
\makeatother

\newcommand{\B}[1]{{\mathbb{#1}}}
\newcommand{\f}[1]{{\boldsymbol{#1}}}

\newcommand{\wha}[1]{\widehat{#1}}
\newcommand{\ba}[1]{{{\bar{#1}}}}
\font\rus=wncyr10 scaled1200

\DeclareMathOperator{\Alt}{{{Alt}}}

\newcommand{\alp}{\alpha}
\newcommand{\gam}{\gamma}
\newcommand{\Gam}{\Gamma}
\newcommand{\del}{\delta}
\newcommand{\sig}{\sigma}
\newcommand{\tht}{\theta}
\newcommand{\lam}{\lambda}
\newcommand{\Lam}{\Lambda}
\newcommand{\ome}{\omega}
\newcommand{\Ome}{\Omega}
\newcommand{\h}{\hbar}

\newcommand{\der}{\partial}
\newcommand{\ten}{\otimes}
\newcommand{\mto}{\mapsto}
\newcommand{\car}{\times}
\newcommand{\wed}{\wedge}
\newcommand{\hto}{\hookrightarrow}
\newcommand{\eqv}{\,\equiv\,}
\newcommand{\con}{\,\lrcorner\,}
\newcommand{\ga}[1]{_0{}^{#1}_0}
\newcommand{\Fla}{^{\flat}{}}
\newcommand{\Sha}{^{\sharp}{}}
\newcommand{\END}{{\,\text{\footnotesize\qedsymbol}}}
\newcommand{\db}[1]{\,{{#1}\!{#1}}\,}
\newcommand{\Rn}{\text{I\!R}}
\newcommand{\byd}{\,{\raisebox{.092ex}{\rm :}{\rm =}}\,}
\newcommand{\M}[1]{{\mathscr{#1}}}
\newcommand{\K}[1]{\text{\rus{#1}}}
\newcommand{\br}[1]{{\breve{#1}}}
\newcommand{\sub}{\subset}
\newcommand{\nab}{\nabla}
\newcommand{\col}[3]{_{#1}{}^{#2}{}_{#3}}
\newcommand{\ssep}[1]{{\qquad\text{\rm{#1}}\qquad}}
\newcommand{\Ga}[2]{_{#1}{}^{#2}_0}
\newcommand{\dt}[1]{{\dot{#1}}}
\newcommand{\1}{\mathbf 1}
\newcommand{\com}{{}\circ{}}
\newcommand{\Ele}{^\mathfrak{e}{}}

\begin{document}
\title[Hidden symmetries]{\textbf{Hidden symmetries of the gravitational contact structure
    of the classical  phase space
    of general relativistic test particle}}
\author[J. Jany\v ska]{Josef Jany\v ska}
\thanks{Supported
by the  grant GA \v CR 14--02476S.}

\address{
{\ }
\newline
Department of Mathematics and Statistics, Masaryk University
\newline
Kotl\'a\v rsk\'a 2, 611 37 Brno, Czech Republic
\newline
e-mail: janyska@math.muni.cz}

\keywords{
phase space,
gravitational contact phase structure,
gravitational Jacobi phase structure,
infinitesimal symmetry,
hidden symmetry,
Killing multivector field}

\subjclass{
70H40, 
70H45, 
70H33, 
70G45, 
58A20. 
}

\begin{abstract}
The phase space of general relativistic test particle is defined as the 
1-jet space of motions. A Lorentzian metric defines  the canonical contact structure on the odd-dimensional phase space. In the paper we study infinitesimal symmetries of the gravitational contact phase structure which are not generated by spacetime infinitesimal symmetries, i.e. they are hidden symmetries. We prove that Killing multivector fields admit hidden symmetries of the gravitational    contact phase structure and we give the explicit description of such hidden symmetries.
\end{abstract}

\maketitle

\section{Introduction}
\setcounter{equation}{0}
A \emph{classical spacetime} is assumed to be an oriented and time oriented 4--dimensional
manifold equipped with a scaled Lorentzian metric.
In classical general relativity the phase space is usually defined either as the cotangent bundle with the canonical symplectic structure or as a part of the unit pseudosphere bundle given by time-like future oriented vectors. In the second case the phase space is also called the observer space and the metric defines the canonical contact phase structure. 

Starting from the papers \cite{JanMod95, Jan14, JanMod08, ManVit04, Vit00} the phase space can be defined as the 1-jet space of motions. In this case we can use the general theory of jets of submanifolds \cite{Vin88} which allows as to define geometrical structures given naturally by the metric and an electromagnetic fields. Namely, the metric field admits the gravitational contact phase structure and the metric and the electromagnetic fields admit the almost-cosymplectic-contact structure \cite{JanMod09}.

In the case of the cotangent bundle a phase infinitesimal symmetry is assumed to be an infinitesimal symmetry of the kinetic energy function. It is very well known, \cite{Cra84, Som73}, that such symmetries are given as the Hamiltonian lift (with respect to the canonical symplectic 2--form) of functions constant of motions.
Functions constant of motions which are polynomial on fibres of the cotangent bundle are given by Killing $ k $--vector fields, $ k\ge 1 $. For $ k= 1 $ the corresponding infinitesimal symmetries are the flow lifts of Killing vector fields and they are projectable on infinitesimal symmetries of the spacetime. For $ k\ge 2 $ the corresponding infinitesimal symmetries are not projectable and they are called {\it hidden symmetries}.

In the case of odd--dimensional phase space (the observer space or the 1--jet space of motions) the metric defines on the phase space a contact structure and a phase infinitesimal symmetry is assumed to be an infinitesimal symmetry of the contact 1--form. Phase infinitesimal symmetries which are projectable on infinitesimal symmetries of the spacetime  were studied on the observer space by Iwai \cite{Iwa76} and on 1--jet space of motions by Jany\v ska and Vitolo \cite{JanVit12}. In both situations projectable symmetries are given by flow lifts of Killing vector fields.
In this paper we describe hidden (nonprojectable) infinitesimal symmetries for the phase space given as the 1-jet space of motions.
It is proved that hidden symmetries are given by the Hamilton--Jacobi lifts of conserved phase functions and we give explicit construction   of hidden symmetries generated by Killing multivector fields. 

\smallskip
 Our theory is explicitly independent of scales, so we
introduce the
spaces of scales
in the sense of \cite{JanModVit10}.
 Any tensor field carries explicit information on its
scale dimension.
We assume the following basic spaces of scales:
the space of 
\emph{time intervals}
$\B T\,$,
the space of 
\emph{lengths}
$\B L $ and
the space of 
\emph{mass} 
$\B M$.
We assume  the 
\emph{speed of light}
$c \in \B T^{*} \ten \B L$ and the 
\emph{Planck constant}
$\hbar\in \B T^*\ten \B L^2\ten \B M$ as the 
\emph{universal scales}.

\section{Preliminaries}
\setcounter{equation}{0}

\subsection{Schouten bracket}
In 1940 Schouten \cite{Sch40} introduced the differential invariant (Schouten bracket) of two contravariant tensor fields (multivector fields). 
 We recall basic facts about the Schouten bracket of skew symmetric and symmetric multivector fields.

Let $ \f M $ be an $ n $-dimensional differentiable manifold and let us denote by $ (x^{\lambda}) $ local coordinates on $ \f M $.
The induced fibred coordinates on $T\f M  $ and $T^{*}\f M  $ will be denoted by $ (x^{\lambda};\overset{\centerdot}{x}{}^{\lambda} ) $ and $ (x^{\lambda};{\overset{\centerdot{}}{x}}{}_{\lambda} ) $, respectively.

Let us recall the expression of the Schouten bracket $ [P,Q] $
 of a skew symmetric 
$p$-vector $P$ and a skew symmetric $q$-vector $Q$ \cite{LibMar87, Vai94}
\begin{equation*}
i_{[P,Q]} \beta = 
(-1)^{pq+q} \, i_P di_Q \beta + (-1)^p \, i_Q di_P \beta \,,
\end{equation*}
for each closed form
$\beta$
of degree
$p+q-1 \,.$ 
Such bracket satisfy the graded antisymmetry and the graded Jacobi identity, so the sheaf of sections $\Gamma(\bigwedge T\f M) = \bigoplus_k \Gamma(\bigwedge^{k}T\f M) $ is a graded Lie algebra.

According to \cite{MicDub95, Woo75} we can define the Schouten bracket for a symmetric $k$-vector field $\overset{k}{K}$ and a symmetric $l$-vector field $\overset{l}{L}$ as the $(k+l-1)$-vector field $[\overset{k}{K},\overset{l}{L}]$ defined for decomposable fields $\overset{k}{K}=X_1\vee\dots \vee X_k$ and  $\overset{l}{L}=Y_1\vee\dots \vee Y_l$ as  
$$ 
[X_1\vee\dots \vee X_k, Y_1\vee\dots \vee Y_l] 
= \sum_{i,j}
[X_i, Y_j ] 
\vee X_{1}\vee\dots \wha{X}_i \dots \vee X_k\vee Y_1 \vee\dots \wha{Y}_j \dots\vee Y_l\,.
$$
The Schouten bracket of symmetric multivector fields is antisymmetric and  satisfy the Jacobi identity,
so $\Gamma(ST\f M) = \bigoplus_k \Gamma(S^{k}T\f M) $ is a Lie algebra with grading but not a graded Lie algebra.

\smallskip
Any symmetric $k$-vector field $\overset{k}{K} = \overset{k}{K}{}^{\lam_{1}\dots\lam_{k}}\,\der_{\lambda_{1}}\ten\dots\ten \der_{\lambda_{k}}$, $\overset{k}{K}{}^{\lam_{1}\dots\lam_{k}}= \overset{k}{K}{}^{\lam_{\sigma(1)}\dots\lam_{\sigma(k)}}  $ for any permutation  of indices $ \sigma $,  defines a function on $T^*\f M$ polynomial and homogeneous of degree $k$ on fibres. 
So we have the mapping
$
\pi^*:\Gamma(S^kT\f M)\to C^\infty(T^*\f M)
$
given in coordinates by
$$
\pi^*(\overset{k}{K}) = \overset{k}{K}{}^{\lam_{1}\dots\lam_{k}}\, {\overset{\centerdot{}}{x}}{}_{\lambda_1}\dots {\overset{\centerdot{}}{x}}{}_{\lambda_k}\,.
$$
$\pi^*$ is a homomorphism of Lie algebras, where on 
$C^\infty(T^*\f M)$ we consider the canonical Poisson bracket $\lbrace , \rbrace$ induced by the canonical symplectic 2-form $ \omega = d{\overset{\centerdot{}}{x}}{}_{\lambda} \wedge dx^\lambda $, \cite{MicDub95}.
I.e. 
$
\pi^*([\overset{k}{K},\overset{l}{L}]) = \lbrace\pi^*(\overset{k}{K}),\pi^*(\overset{l}{L})\rbrace
$
which gives the following coordinate expression, 
for a symmetric $k$-vector field $\overset{k}{K}$ and a symmetric  $l$-vector field $\overset{l}{L}$, 
\begin{align}\label{Eq: 2.1}
[\overset{k}{K},\overset{l}{L}] 
& = 
\frac{1}{(k+l-1)!} \sum_\sigma \big(
k\, \overset{k}{K}{}^{\rho\lam_{\sigma(1)}\dots\lam_{\sigma(k-1)}}\,\der_\rho \overset{l}{L}{}^{\lam_{\sigma(k)}\dots\lam_{\sigma(k+l-1)}}
\\
&\quad\nonumber
- l\, \overset{l}{L}{}^{\rho\lam_{\sigma(1)}\dots\lam_{\sigma(l-1)}}\,\der_\rho \overset{k}{K}{}^{\lam_{\sigma(l)}\dots\lam_{\sigma(k+l-1)}}
\big)\, \der_{\lambda_1}\ten \dots \ten \der_{\lambda_{k+l-1}}\,,
\end{align}
where $\sigma$ runs all permutations of indices $1, \dots , k+l-1$.

\subsection{Killing tensor fields}

We recall basic facts about Killing tensor fields on a Riemannian or a pseudo-Riemannian manifold $ (\f M,g) $. In what follows we shall denote by $\ba g$ the contravariant metric.

\smallskip
A {\it Killing $ (0,k) $-tensor field} is a symmetric $ (0,k) $-tensor field $K=K_{\lambda_{1}\dots\lambda_{k}} \, d^{\lambda_{1}}\ten\dots\ten d^{\lambda_{k}} $ such that the function $ K_{\lambda_{1}\dots\lambda_{k}} \, \overset{\centerdot}{x}{}^{\lambda_{1}}\dots \overset{\centerdot}{x}{}^{\lambda_{k}}$ is constant on geodesic curves of the Levi Civita connection. This condition is equivalent with
\begin{equation*}
\nabla_{(\lambda_{1}}K_{\lambda_{2} \dots \lambda_{k+1})}  = 0\,.
\end{equation*}

Equivalently we can define a {\it Killing $ (k,0) $-tensor field} ($ k $-vector field) as  a symmetric $ (k,0) $-tensor field $\overset{k}{K}=\overset{k}{K}{}^{\lambda_{1}\dots\lambda_{k}} \, \der_{\lambda_{1}}\ten\dots\ten \der_{\lambda_{k}} $ such that 
\begin{equation*}
g^{\rho(\lambda_{1}}\nabla_{\rho}\overset{k}{K}{}^{\lambda_{2}\dots\lambda_{k+1})} =
\nabla^{(\lambda_{1}}\overset{k}{K}{}^{\lambda_{2}\dots\lambda_{k+1})} = 0\,.
\end{equation*}

It is easy to see that  $ K $ is a Killing $ (0,k) $-tensor field if and only if $ K\Sha = (g\Sha\ten\dots\ten g\Sha) K $ is a Killing $ (k,0) $-tensor field.

\begin{rmk}
{\rm 
We have the canonical Killing (0,2) and (2,0) tensor fields given by  $g$ and $\bar g$, respectively.
}
\hfill\END
\end{rmk}

Killing multivector fields can be equivalently defined as
symmetric multivector fields satisfying the {\it Killing tensor equation}
\begin{equation}\label{Eq: 2.2}
[\overset{k}{K},\ba g] = 0\,,
\end{equation}
 \cite{Woo75}.
Then from the Jacobi identity we get that the sheaf of Killing multivector fields is closed with respect to the Schouten bracket, i.e. if $\overset{k}{K},\overset{l}{L}$ are Killing multivector fields then $[\overset{k}{K},\overset{l}{L}]$ is also a Killing multivector field.

\begin{rmk}
{\rm
In \cite{Som73} it was proved that for a symmetric $k$--vector field $\overset{k}{K}$ the function $\pi^*(\overset{k}{K})$ is constant of motion if and only if $\overset{k}{K}$ is  Killing. On the other hand it is equivalent with the fact that the Hamiltonian lift of $\pi^*(\overset{k}{K})$ given by the canonical symplectic 2--form is an infinitesimal symmetry of the kinetic energy function $ \tfrac 12 \bar g^{\lam\mu}\,\dot x_\lambda\, \dot x_\mu $.  
}
\hfill\END
\end{rmk}

\subsection{Structures of odd dimensional manifolds}
Let $\f M$ be a $(2n+1)$-dimensional manifold.

A
\emph{pre cosymplectic (regular) structure (pair)} 
on $\f M$ is given by a 1-form $\ome$ and a 2-form $\Ome$
such that
$
\ome\wed \Ome^n \not\equiv 0\,.
$
A \emph{contravariant (regular) structure (pair)} $(E,\Lam)$ is given by
a vector field $E$ and an antisymmetric 2-vector $\Lam$ such that
$
E\wed \Lam^n \not\equiv 0\,.
$
We denote by
$
\Ome\Fla:T\f M\to T^*\f M\,
$
and
$
\Lam\Sha:T^*\f M\to T\f M\,
$
the corresponding "musical" morphisms.

By \cite{Lich78}
if $(\ome,\Ome)$ is a pre cosymplectic pair then
there exists a unique regular pair
$(E,\Lam)$ such that
\begin{equation}\label{Eq: 2.3}
(\Ome\Fla_{|\text{im}\, \Lam\Sha})^{-1}
= \Lam\Sha_{|\text{im}\, \Ome\Fla} \,,
\quad i_E\ome =1\,,
\quad i_E\Ome =0\,,
\quad i_\ome\Lam = 0\,.
\end{equation}
On the other hand for any regular pair $(E,\Lam)$
there exists a unique (regular) pair $(\ome,\Ome)$
satisfying the above identities.
The pairs $(\ome,\Ome)$ and
$(E,\Lam)$ satisfying the above identities are said to be mutually 
\emph{dual}.
The vector field $E$ is usually called
the 
\emph{Reeb vector field}
of the pair $(\ome,\Ome)$. In fact
geometrical structures given by dual pairs coincide. 

An 
\emph{almost-cosymplectic-contact structure (pair)}
\cite{JanMod09}
is
given by
a pair 
$(\ome,\Ome)$
such that
$
d\Ome = 0 \,,
$
$
\ome \wed \Ome^n \not\equiv 0 \,.
$
The dual 
\emph{almost-coPoisson-Jacobi structure (pair)\/} 
is given by the
pair $(E,\Lam)$
such that
$
[E,\Lam] = - E\wed \Lam\Sha(L_E\ome)\,,
$
$
[\Lam,\Lam] = 2 \, E\wed (\Lam\Sha\ten\Lam\Sha)(d\ome)\,,
$
where $\ome$ is the 
\emph{fundamental 1-form\/}
satisfying $i_E\ome=1\,,\,\, i_\ome\Lam =0$.
Here $[,]$ is the Schouten bracket of skew symmetric multivector fields.

\begin{rmk}\label{Rm: 2.2}
{\rm
An almost-cosypmlectic-contact pair generalizes standard cosymplectic and contact pairs.
Really, if $d\ome =0$ we obtain a cosymplectic pair \cite{deLTuy96}.
The corresponding dual pair is \emph{coPoisson pair} \cite{JanMod09}
given by the pair
$(E, \Lam)$
such that
$
[E, \Lam] = 0 \,,
$
$
[\Lam, \Lam] = 0 \,.
$
A 
{contact structure (pair)}
is given by a pair 
$(\ome,\Ome)$
such that
$
\Ome = d\ome \,,
$
$
\ome \wed \Ome^n \not\equiv 0 \,.
$
The dual
{Jacobi structure (pair)}
is given by the pair
$(E, \Lam)$
such that
$
[E, \Lam] = 0 \,,
$
$
[\Lam, \Lam] = - 2 E \wed \Lam \,.
$
}
\hfill\END
\end{rmk}

\subsection{Infinitesimal symmetries of almost-cosymplectic-contact structures}

Let $(\ome,\Ome)$ and $(E,\Lambda)$ be mutually dual regular structures on an odd dimensional manifold $\f M$. An \emph{infinitesimal symmetry}  of the structure $(\ome,\Ome)$ is a vector
field $ X $ on $\f M $ such that $L_X\omega = 0$, $L_X\Omega = 0$.
Similarly, an \emph{infinitesimal symmetry}  of the structure $(E,\Lambda)$ is
a vector field $X$ on $\f M$ such that
$L_X E = [X,E] = 0$, $L_X\Lambda = [X,\Lambda] =  0$.

\begin{lem}\label{Lm: 2.1} {\rm (\cite{JanVit12})}
Let $X$ be a vector field on $\f M$. The following conditions
are equivalent:

1. $L_X\omega = 0$ and $L_X\Ome = 0$.

2. $L_X E= [X,E]=0$ and
$L_X\Lambda= [X,\Lambda]=0$.
\hfill\END
\end{lem}

\begin{thm}\label{Th: 2.2} {\rm (\cite{Jan11})}
{\it
A vector field $ X $ is an infinitesimal symmetry of the almost-cosymplectic-contact structure $(\ome,\Ome)$
if and only if
it is of local type
$
X = df\Sha + h\, E\,,
$
where $ f,h \in C^{\infty}(\f M) $ such that $ E.f = 0 $ and
\begin{equation}\label{Eq: 2.4}
i_{df\Sha}d\omega + h\, i_E\, d\omega + dh = 0\,.
\end{equation}
}
\end{thm}

\begin{crl}\label{Cr: 2.3}
1. An infinitesimal symmetry
of the cosymplectic  structure $(\omega,\Omega)$ is of local type
$
Y = df\Sha + h\, E\,,
$
where $ f \in C^\infty(\f M) $ such that $ E.f = 0 $ and $h$ is a constant. 

2. Any infinitesimal symmetry
of the contact  structure $(\omega,\Omega)$ is
 of local type
\begin{equation}\label{Eq: 2.5}
X =  df\Sha - f\,E\,,
\end{equation} 
where $ f \in C^\infty(\f M) $ such that $ E.f = 0 $.
\end{crl}

\begin{proof}
1. For a cosymplectic structure we have $d\omega=0$ and \eqref{Eq: 2.4} reduces to
$
dh = 0\,.
$

2. For a contact structure we have $d\ome = \Omega$ and \eqref{Eq: 2.4} reduces to
$
i_{df\Sha}  \Omega + dh = 0\,,
$
i.e.
$
dh = - df\,.
$ Then $h=-(f+k)$, where $k$ is a constant.
The vector field \eqref{Eq: 2.5} is the {\it Hamilton-Jacobi lift} of a function $f$. But the Hamilton-Jacobi lift of a constant $k$ is the vector field $k\, E$ which is an infinitesimal symmetry of the contact structure $(\omega, d\omega)$.
So all infinitesimal symmetries of the contact structure
form an $\mathbb{R}$-algebra and they are Hamilton-Jacobi lifts of functions on $\f M$ satisfying $E.f =0$.
\end{proof}

\begin{lem}\label{Lm: 2.4}
We have 
\begin{equation*}
i_E dh + i_E i_{df\Sha}d\omega
= E.h + \Lambda(L_E\ome,df) = 0\,.
\end{equation*}
So $dh +  i_{df\Sha}d\omega\in\ker E$.
\end{lem}

\begin{proof}
If we apply $i_E$ on the equation  \eqref{Eq: 2.4} we get 
$ i_E dh + i_E i_{df\Sha}d\omega
 = E. h - i_{df\Sha}i_E d\omega 
 = E.h - \Lam(df,L_E\omega)
= 0 $.
\end{proof}

\subsection{Lie bracket of generators of infinitesimal symmetries}

Any infinitesimal symmetry of the almost-cosymplectic-contact structure $(\ome,\Ome)$ can be identified with a pair of functions $(f,h)$ on $\f M$ such that $f$ is {\it conserved}, i.e. $E.f = 0$,
and $f$ and $h$ are related by the condition \eqref{Eq: 2.4}. The pair $ (f,h) $ is said to be a {\it generator} of the infinitesimal symmetry of the almost-cosymplectic-contact structure $(\ome,\Ome)$.

\begin{lem}\label{Lm: 2.5}
Suppose two infinitesimal symmetries 
$
X = df\Sha + h\, E\,
$
and 
$
X' = dg\Sha + k\, E\,
$
of the almost-cosymplectic-contact structure $(\ome,\Ome)$.
Then
\begin{align*}
[X,X'] 
& =
d\{f,g\}\Sha +\big(
\{f,k\} - \{g,h\} - d\omega(df\Sha,dg\Sha)
\big)\, E\,.
\end{align*}
\end{lem}

\begin{proof}
We have
$$
[X,X'] = [df\Sha,dg\Sha] + [df\Sha,k\, E] + [h\, E,dg\Sha]
+ [h \, E,k\, E].
$$
By \cite{JanMod09} we have
\begin{align*}
[df\Sha,dg\Sha]
& =
d\Lambda(df,dg)\Sha
- d\omega(df\Sha,dg\Sha) \, E
= d\{f,g\}\Sha
- d\omega(df\Sha,dg\Sha)\, E\,,
\\
[E,df\Sha]
& =
\Lambda(L_E\omega,df)\, E 
= i_Ei_{df\Sha}d\omega\, E
\end{align*}
which implies
\begin{align*}
[X,X'] 
& =
d\{f,g\}\Sha +\big(
\{f,k\} - \{g,h\} - d\omega(df\Sha,dg\Sha)
\\
& \quad
+ h\,(E.k + \Lambda(L_E\omega,dg))
- k\,(E.h + \Lambda(L_E\omega,df)) 
\big)\, E\,
\end{align*}
and from Lemma \ref{Lm: 2.4} we obtain Lemma \ref{Lm: 2.5}.
\end{proof}

Infinitesimal symmetries form a Lie algebra with respect to the Lie bracket. 
This defines the Lie bracket on pairs of functions
\begin{equation}\label{Eq: 2.6}
\db[(f,h);(g,k)\db]
=
\big(
\{f,g\}, \{f,k\} - \{g,h\} - d\omega(df\Sha,dg\Sha)
\big)\,.
\end{equation}

Really, this bracket is antisymmetric and satisfies the conditions for generators of infinitesimal symmetries. Namely, from  \cite{JanMod09} and properties of the almost-coPoisson-Jacobi structure we have
\begin{equation}\label{Eq: 2.7}
E.\{f,g\} = \{E.f,g\} + \{f,E.g\} + i_{[E,\Lambda]} df \wedge dg = 0\,,
\end{equation}
i.e. the sheaf of functions satisfying $E.f=0$ is closed with respect to the Poisson bracket.  

The condition \eqref{Eq: 2.4} corresponds to
\begin{align*}
i_{d\{f,g\}\Sha} d\omega 
 +
\big(
\{f,k\} - \{g,h\} - d\omega(df\Sha,dg\Sha)
\big)\, i_E d\omega
&
\\
+ 
d\{f,k\} - d\{g,h\} - d(d\omega(df\Sha,dg\Sha))
& = 
0
\end{align*}
which follows from 
$
L_{[X,X']} \omega = (L_XL_{X'} - L_{X'}L_X)\,\omega\,.
$

\begin{crl}\label{Cr: 2.6}
1. For a cosymplectic structure $h, k$ are constants and the 
bracket \eqref{Eq: 2.6} is reduced to
\begin{equation*}
\db[(f,h);(g,k)\db]
=
\big(
\{f,g\}, 0
\big)\,.
\end{equation*}

2. For a contact structure  
the bracket \eqref{Eq: 2.6} is reduced to
\begin{equation*}
\db[(f,-f);(g,-g)\db]
=
\big(
\{f,g\}, \{f,-g\}- \{g,-f\}- \Omega(df\Sha,dg\Sha)
\big)
=
\big(
\{f,g\}, - \{f,g\}
\big)
\,.
\end{equation*}
\end{crl}

\section{Structures on the classical phase space}
\setcounter{equation}{0}
 Then, we study the geometrical structures arising on the phase
space of a classical spacetime
\cite{JanMod95,JanMod08}.

\subsection{Einstein spacetime}
 We assume \emph{spacetime} to be an oriented 4--dimensional
manifold
$\f E$
equipped with a scaled Lorentzian metric
$g : \f E \to \B L^2 \ten (T^*\f E \ten T^*\f E) \,,$
with signature
$(-+++) \,;$
we suppose spacetime to be time oriented.
 The contravariant metric is denoted by
$\ba g : \f E \to \B L^{-2} \ten (T\f E \ten T\f E) \,.$

 A \emph{spacetime chart} is defined to be a chart
$(x^\lam) \equiv (x^0, x^i) \in C^\infty(\f U, \, \Rn \car \Rn^3)$, $ \f U\subset \f E $ is open,
of
$\f E \,,$
which fits the orientation of spacetime and such that the vector field
$\der_0$
is timelike and time oriented and the vector fields
$\der_1, \der_2, \der_3$
are spacelike.
 Greek indices
$\lam, \mu, \dots$
will span spacetime coordinates, while Latin indices
$i, j, \dots$
will span spacelike coordinates.
 In the following, we shall always refer to spacetime charts.
 The induced local bases of
$T\f E$
and
$T^*\f E$
are denoted, respectively, by
$(\der_\lam)$
and
$(d^\lam) \,.$
 We have the coordinate expressions
\begin{alignat*}{3}
g
&=
g_{\lam\mu} \, d^\lam \ten d^\mu \,,
&&\ssep{with}
g_{\lam\mu}= g_{\mu\lambda}
&&\in C^\infty(\f E, \, \B L^2 \ten \Rn) \,,
\\
\ba g
&=
g^{\lam\mu} = g^{\mu\lam} \, \der_\lam \ten \der_\mu\,,
&&\ssep{with}
g^{\lam\mu}
&&\in C^\infty(\f E, \, \B L^{-2} \ten \Rn) \,.
\end{alignat*}
For a particle with a mass $m$ it is very convenient
to use the 
\emph{
re-scaled metric}
$G = \frac m\hbar g :
\f E \to \B T \ten (T^*\f E \ten T^*\f E)$, $G^{0}_{\lam\mu} = \tfrac{m}{\h_0}\,g_{\lam\mu} $, and 
the associated contravariant re-scaled metric
$\ba G = \frac \hbar m \ba g :
\f E \to \B T^* \ten (T\f E \ten T\f E) \,,$ $G_{0}^{\lam\mu} = \tfrac{\h_0}{m}\,g^{\lam\mu} $, where $ \h = \h_0\, u^{0}, \,\, \h_0 \in \B L^{2}\ten \B M $. 
Eventually, we consider the {\it unscaled metric} $\wha G \byd \big(\tfrac{mc}{\h} \big)^2 g:\f E \to T^*\f E \ten T^*\f E$, $\wha G_{\lam\mu} = (\tfrac{m\,c}{\h})^2\,g_{\lam\mu} $, and the associated contravariant unscaled metric $\wha{\ba G} \byd \big(\tfrac{\h}{m\,c}\big)^2 \ba g:\f E \to T\f E \ten T\f E$, $\wha G^{\lam\mu} = (\tfrac{\h}{m\,c})^2\,g^{\lam\mu} $.

\subsection{Phase space}
\label{E: Phase space}

Our theory is based on the theory of $r$-jets of $k$-dimensional 
submanifolds of a manifold $\f M$ denoted by $J_r(\f M,k)$. 
In literature $J_r(\f M,k)$ is known also as the space of $r$-th 
order contact elements of dimension $k$ \cite{Vin88}.

We assume \emph{time}  to be a one-dimensional affine space $ \f T $
associated with the vector space $ \bar{\B T}= \B T \otimes \B R $.
A \emph{motion} is defined to be a 1--dimensional timelike
submanifold
$s : \f T \hto \f E \,.$
The \emph{1st differential} of the motion
$s$
is defined to be the tangent map $ ds:T\f T = \f T \times \bar{\B T} \to T\f E $.

\smallskip

 We assume as \emph{phase space} the open subspace 
$\M J_1\f E \sub J_1(\f E,1)$
consisting of all 1--jets of motions.  The \emph{velocity} of a motion
$s$
is defined to be its 1--jet
$j_1s : \f T \to \M J_1(\f E,1) \,.$
 For each 1--dimensional submanifold
$s : \f T \hto \f E$
and for each
$x \in \f T \,,$
we have
$j_1s(x) \in \M J_1\f E$
if and only if
$ds(x)(u)\in  T_{s(x)}\f E$
is timelike, where $ u \in {\B T} $.

 Any spacetime chart
$(x^0, x^i)$
is related to each motion
$s \,$ which means that $ s $ can be locally expressed by $ (x^{0}, x^{i}= s^{i}(x^{0})) $. 
Then we obtain the induced fibred coordinate chart $(x^0, x^i, x^i_0)$
 on
$\M J_1\f E$ such that $ x^{i}_0\com s = \der_0 s^{i} $.
Moreover, there exists a time unit function $ \f T\to \B T $ such that the 1st differential of $ s $, considered as the
map
$ds  : \f T \to \bar{\B T}^* \ten T\f E \,$,
is normalized by $g(ds, \, ds) = - c^2 \,$, for details see \cite{JanMod08}.

We shall always refer to the above fibred charts.

\smallskip

 We define the \emph{contact map} to be the unique fibred morphism
$\K d : \M J_1\f E \to \bar{\B T}^* \ten T\f E$
over
$\f E \,,$
such that
$\K d \com j_1s = ds \,,$
for each motion
$s \,.$
 We have
$g \, (\K d, \K d) = - c^2 \,.$
 The coordinate expression of
$\K d$
is
\begin{equation}\label{Eq: 3.1}
\K d =
c \, \alp^0 \,  (\der_0 + x^i_0 \, \der_i) \,,
\ssep{where}
\alp^0 \byd
1 /\sqrt{|g_{00} + 2 \, g_{0j} \, x^j_0 + g_{ij} \, x^i_0 \, x^j_0|}
\,.
\end{equation}

 The map
$\K d : \M J_1\f E \to \B T^*\ten T\f E$
is injective.
 Indeed, it makes
$\M J_1\f E \sub \bar{\B T}^*\ten T\f E$
the fibred submanifold over
$\f E$
characterised by the constraint
$g_{\lam\mu} \, \dt x^\lam_0 \, \dt x^\mu_0 = - (c_0)^2 \,.$

 We define the 
\emph{time form} 
to be the fibred morphism
$
\tau
= - \frac1{c^2}\, g\Fla({\K d}) :
{\M J}_1{\f E} \to {\B T}\ten T^*{\f E} \,,
$
considered as the scaled
horizontal 1--form of $ \M J_1\f E $.
We have the coordinate expression
\begin{equation}\label{Eq: 3.2}
\tau = \tau_\lam\, d^\lam = - \tfrac{\alp^0}{c}
\, (g_{0\lam}+ g_{i\lam}\, x^i_0) \, d^\lam \,.
\end{equation}
The 
\emph{complementary contact map}
$
\tht:
{\M J}_1\f E \times_{\f E} T\f E \to T\f E \,
$
is given by
$
\tht = \mathrm{id} - {\K d}\ten \tau
 \,.
$

\begin{nt}\label{Nt: 3.1}
In what follows it is very convenient to use the following notation $\br\del^i_\lam = \del^i_\lam - x^i_0 
\, \del^0_\lam$ and $\br\del^\mu_0 = \del^\mu_0 +  
 \del^\mu_p\, x^p_0$\,. Then ${\K d} = c\,\alpha^{0}\,\br\del^\mu_0\, \der_\mu   $ and $ \tau =  - \tfrac{\alp^0}{c}
\, \br g_{0\lam}\, d^\lam \,, $ where $ \br g_{0\lam} = g_{\mu\lambda}\, \br\del^\mu_0\,. $
\hfill\END
\end{nt}

 Let
$V\M J_1\f E \sub T\M J_1\f E$
be the vertical tangent subbundle over
$\f E \,.$
 The vertical prolongation of the contact map yields the mutually
inverse linear fibred isomorphisms
\begin{equation*}
\nu_\tau : \B T^{*} \ten V_\tau\f E \to V\M J_1\f E
\ssep{and}
\nu^{-1}_\tau :
V\M J_1\f E  \to  \B T^* \ten V_\tau\f E \,,
\end{equation*}
where $V_\tau\f E = \ker \tau\subset T\f E$, with the coordinate expressions
\begin{equation}\label{Eq: 3.3}
\nu_\tau =  \frac1{c \, \alp^0} \, \br\del^i_{\lambda}\, d^\lambda \ten \der^0_i\,,
\quad
\nu^{-1}_\tau = c \, \alp^0 \, d^i_0 \ten 
\big(\der_i - c\,\alp^0\tau_i\,\br\delta^{\lambda}_0\,\der_\lambda\big) \,,
\end{equation}
where $ \der^0_i=\tfrac{\der}{\der x^{i}_0} $ and $ d^{i}_0 = dx^{i}_0 $.

\subsection{Spacetime and phase connections}
\label{E: Spacetime connections}
 We define a \emph{spacetime connection} to be a torsion free linear
connection 
$K : T\f E \to T^*\f E \ten TT\f E$
of the bundle
$T\f E \to \f E \,.$
 Its coordinate expression is of the type
$
K =
d^\lam \ten 
(\der_\lam + K\col\lam\nu\mu \, \dot x^\mu \, \dt\der_\nu) \,,
$
{with}
$
K\col\mu\nu\lam = K\col\lam\nu\mu \in C^\infty(\f E) \,.
$
\smallskip

  We denote by 
$K[g]$ 
the \emph{Levi Civita connection}, i.e. the torsion free
linear spacetime connection such that 
$\nab g=0 \,.$

 We define a \emph{phase connection} to be a connection of the
bundle
$\M J_1\f E \to \f E \,.$
A phase connection can be represented, equivalently, by a tangent
valued form
$\Gam : \M J_1\f E \to T^*\f E \ten T\M J_1\f E \,,$
which is projectable over
$\1 : \f E \to T^*\f E \ten T\f E \,,$
or by the complementary vertical valued form
$\nu[\Gam] :
\M J_1\f E \to T^*\M J_1\f E \ten V\M J_1\f E \,,$
or by the vector valued form
$\nu_\tau[\Gam] \byd \nu^{-1}_\tau \com \nu[\Gam] : \M J_1\f E \to
T^*\M J_1\f E \ten (\B T^* \ten V_\tau\f E) \,.$
 Their coordinate expressions are
\begin{gather*}
\Gam =
d^\lam \ten(\der_\lam + \Gam\Ga\lam i \, \der_i^0) \,,
\qquad
\nu[\Gam] =
(d^i_0 - \Gam\Ga\lam i \, d^\lam) \ten \der^0_i \,,
\\
\nu_\tau[\Gam] =
c \, \alp^0 \, (d^i_0 - \Gam\Ga\lam i \, d^\lam) \ten 
	\big(\der_i - c\,\alp^0\tau_i(\der_0+x^p_0\,\der_p)\big)\,,
\ssep{with}
\Gam\Ga\lam i \in C^\infty(\M J_1\f E) \,.
\end{gather*}

 We can prove \cite{JanMod08} that there is a natural map 
$\chi : K \mto \Gam$
between linear spacetime connections 
$K$ 
and phase connections 
$\Gam \,,$
with the coordinate expression
$
\Gam\Ga \lam i = 
\br \del^i_\rho\, K_\lam{}^\rho{}_\sigma\, \br\del^\sigma_0 \,.
$
\subsection{Dynamical phase connection}
\label{E: Dynamical phase connection}
 The space of 2--jets of motions
$\M J_2\f E$
can be naturally regarded as the affine subbundle
$\M J_2\f E \sub \B T^* \ten T\M J_1\f E \,,$
which projects on
$\K d : \M J_1\f E \to \B T^* \ten T\f E \,.$

 A \emph{dynamical phase connection} is defined to be a 2nd--order
connection, i.e. a section
$\gam : \M J_1\f E \to \M J_2\f E \,,$
or, equivalently, a section
$\gam : \M J_1\f E \to \B T^* \ten T\M{J}_1\f E \,,$
which projects on
$\K d \,.$

 The coordinate expression of a dynamical phase connection is of the
type
\begin{equation}\label{Eq: 3.4}
\gam = c \, \alp^0 \,
(\der_0 + x^i_0 \, \der_i + \gam\ga i \, \der^0_i) \,,
\ssep{with}
\gam\ga i \in C^\infty(\M J_1\f E) \,.
\end{equation}
 If
$\gam$
is a dynamical phase connection, then we have
$\gam \con \tau = 1 \,.$

 The contact map
$\K d$
and a phase connection
$\Gam$
yield the section
$\gam \eqv \gam[\K d, \Gam] \byd 
\K d \con \Gam : \M J_1\f E \to \B T^* \ten T\M J_1\f E \,,$
which turns out to be a dynamical phase connection, with the coordinate
expression
$
\gam\ga i = \Gam\Ga \rho i\, \br\delta^\rho_0 \,.
$
In particular, a linear spacetime connection
$K$
yields the dynamical phase connection
$\gam \byd \gam[\K d, K] \byd \K d \con \chi(K) \,,$
with the coordinate expression
$
\gam\ga i 
=
\br\del^i_\rho\, K_\sigma{}^\rho{}_\tau\, \br\delta^\sigma_0\, \br\delta^\tau_0 \,.
$
For the Levi Civita connection we get the {\it gravitational dynamical phase connection}
$ \gamma[g] = \K{d}\con\chi(K[g]). $

\subsection{Phase 2--form and 2--vector}
\label{E: Phase 2--form and 2--vector}
 The rescaled metric
$G$
and a phase connection
$\Gam$
yield the 2--form 
$\Ome \,,$ 
called \emph{phase
2--form}, and  the vertical 2--vector 
$\Lam \,,$ 
called \emph{phase 2--vector}, 
\begin{alignat}{3}
\Ome \label{Eq: 3.5}
&\byd 
\Ome[G,\Gam] 
&&\byd
G \con \big(\nu_\tau[\Gam]\wed\tht\big) 
&&:
\M J_1\f E \to \bigwedge^2T^*\M J_1\f E \,,
\\
\Lam \label{Eq: 3.6}
&\byd 
\Lam[G, \Gam] 
&&\byd
\ba G \con (\Gam \wed \nu_\tau) 
&&:
\M J_1\f E \to \bigwedge^2T\M J_1\f E \,,
\end{alignat}
with the coordinate expressions
\begin{equation}\label{Eq: 3.7}
\Ome =
c_0 \, \alp^0 \, \br G^0_{i\mu} \,
(d^i_0 - \Gam\Ga \lam i \, d^\lam) \wed d^\mu\,,
\quad
\Lam =
\frac1{c_0 \, \alp^0} \, \br G^{j\lam}_0 \,
(\der_\lam + \Gam\Ga \lam i \, \der^0_i) \wed \der^0_j \,,
\end{equation}
where $ \br G^0_{i\mu} = G^0_{i\mu} + (\alpha^0)^2\, g_{\rho i}\, G^0_{\sigma\mu}\,\br\delta^\rho_0\,\br\delta^\sigma_0 $ and $ \br G^{i\lam}_0 =\br\delta^i_\rho\,  G^{\rho\lam}_0$.

 We can easily see that  
$\tau \wed \Ome^3 \not\equiv 0$ 
and
$\gam \wed \Lam^3 \not\equiv 0 \,.$ 

 There is a unique dynamical phase connection
$\gam \,,$
such that
$\gam \con \Ome[g, \Gam] = 0 \,.$
Namely,
$\gam = \gam[\K d, \Gam] \,.$

 In particular, a metric and time preserving spacetime connection
$K$
yields the phase 2--form 
$\Ome[G,K]\byd \Ome[G,\chi(K)]$ 
and the  phase 2--vector
$\Lam[G,K]\byd \Lam[G,\chi(K)]$.
Moreover, for the Levi Civita connection we get the gravitational  phase 2--form 
$\Ome[g]\byd \Ome[G,\chi(K[g])]$ 
and the  phase 2--vector
$\Lam[g]\byd \Lam[G,\chi(K[g])]$.

\subsection{Electromagnetic structure}


Now, we assume the \emph{electromagnetic field} to be a closed scaled
$2$-form on $\f E$
\begin{equation}\label{Eq: 3.8}
  F :\f E \to (\B L^{1/2}\otimes\B M^{1/2}) \otimes \bigwedge^2
  T^*\f E.
\end{equation}

Given a charge
$q$,
the rescaled electromagnetic field
$\wha F = (q/2\hbar)\,F$
can be incorporated into the geometrical structure of the phase space,
\emph{i.e.}\ the gravitational 2--form.  Namely, we define the
\emph{joined (total) phase 2-form}
\begin{equation*}
  \Omega \byd \Omega[g] + \frac{q}{2\hbar} F = \Omega[g] + \Omega\Ele:\M J_1\f E \to 
  \bigwedge^2
   T^*\M J_1\f E.
\end{equation*}
Of course $d \Omega = 0$ but $\Omega$ is exact if and only if $F$ is
exact.

We recall \cite{JanMod08} that a unique connection
$\Gamma$
on
$\M J_1\f E\to \f E$
can be characterized through the total 2-form
$\Omega$
 by the formula \eqref{Eq: 3.5}. Namely the {\em joined (total) phase connection}
$
\Gam = \Gam[g] + \Gam\Ele\,,
$
where
$$\Gam\Ele \byd
- \frac{1}{2} \big(\nu_\tau \com G\Sha^2\big)\big(\wha F + 2\, \tau \wed (\K d\con\wha F)\big) \,
$$
with the coordinate expression
$\Gam\Ele =
- (1/(2c_0\alp^0)) \br G^{i\mu}_0 \,
(\wha F_{\lam\mu} - (\alp^0)^2 \br g_{0\lam} \,
\wha F_{\rho\mu} \, \br\del^\rho_0) \, d^\lam \ten \der^0_i$, here $ G\Sha{}^{2} $ means that we apply $ G\Sha{} $ on the second index.

The total phase connection then admits the {\it joined (total) phase 2--vector}
$
\Lambda = \Lambda[g] + \Lambda\Ele
$
and the {\it joined (total) dynamical connection}
$
\gam = \gamma[g] + \gam\Ele\,
$
given by
$$
\Lam\Ele =
\tfrac{1}{2} 
\Alt\big((\nu_\tau\com g\Sha) \ten (\nu_\tau\com g\Sha)\big) (\widehat F)
\, \qquad
\mathrm{and}
\qquad
\gam\Ele =:{\K d}\, \con \, \Gamma\Ele: {\M J}_1\f E 
\to \B T^*\ten V{\M J}_1\f E\,.
$$
Here $ \gam\Ele $
is the Lorentz force.

\subsection{Dynamical structures of the phase space}
\label{E: Dynamical structures of the phase space}

In what follows we shall use the {\it unscaled time form} 
$ \wha\tau = \tfrac{m\,c^2}{\h}\, \tau $. 
First, let us consider the gravitational objects $\wha\gam[g] = \tfrac{\h}{m\,c^2}\, \gamma[g]$, $\Omega[g] $ and $\Lam[g] $.
 
\begin{thm}\label{Th: 3.1} {\rm (\cite{JanMod08})}
 We have: 

\smallskip

{\rm (1)} 
$\Ome[g]= - d\wha\tau \,,$ 
i.e. 
$(- \wha \tau, \Ome[g])$
is a contact pair.

\smallskip

{\rm (2)} 
$\big[\wha\gam[g], \Lam[g]\big] = 0$ 
and 
$\big[\Lam[g], \Lam[g]\big] =  2\,\wha\gam[g] \wed \Lam[g] \,,$
i.e. 
$(- \wha\gam[g], \Lam[g])$ 
is a (regular) Jacobi pair.

\smallskip

 Moreover, the contact pair 
$( - \wha\tau, \Ome[g])$ 
and
the (regular) Jacobi pair 
$(-\wha\gam[g], \Lam[g])$ 
are mutually dual.
\hfill\END
\end{thm}

According to Theorem \ref{Th: 3.1} the metric $ g $ defines on the phase space the natural contact structure which will be called the {\it gravitational contact phase structure}. Dually the metric defines the gravitational Jacobi phase structure. 

\begin{rmk}\label{Rm: 3.1}
{\rm
Let us remark that in the standard literature the phase space (the observer space) is defined as a part of the unit pseudosphere bundle formed by timelike and future oriented vectors. The contact 1-form is then obtained either as the restriction of the 1-form $ \alpha = g_{\lambda\mu}\,\dt x^\lambda\, d^\mu $ to the phase space \cite{Iwa76} or by a direct construction described in \cite{GieWis13}.   
}
\hfill\END
\end{rmk}

Further let us consider the joint objects
$\wha\gam = \wha\gam[g] + \wha\gamma\Ele$, $\Omega=\Omega[g]+\Omega\Ele $ and $\Lam = \Lam[g]+ \Lambda\Ele $. 

\begin{thm}\label{Th: 3.2} {\rm (\cite{JanMod08})}
 We have: 

\smallskip

{\rm (1)} 
$\wha\tau\wedge\Omega^3 \not\equiv 0  $,  $d\Ome= 0 \,,$ 
i.e. 
$(- \wha \tau, \Ome)$
is an almost-cosymplectic-contact pair.

\smallskip

{\rm (2)} 
$[\wha\gam, \Lam] = \wha\gamma\wedge\Lambda\Sha(L_{\wha\gamma}\wha\tau)$ 
and 
$[\Lam, \Lam] = - 2\,\wha\gam \wed (\Lam\Sha\ten\Lambda\Sha)(d\wha\tau) \,,$
i.e. 
$(- \wha\gamma, \Lam)$ 
is a (regular) almost-coPoisson-Jacobi pair.

\smallskip

 Moreover, the pairs 
$( - \wha\tau, \Ome)$ 
and 
$(-\wha\gam, \Lam)$ 
are mutually dual.
\hfill\END
\end{thm}
 
According to Theorem \ref{Th: 3.2} the metric $ g $ and the electromagnetic 2-form $ F $ define on the phase space the almost-cosymplectic-contact structure which will be called the {\it joined (total) almost-cosymplectic-contact phase structure}. Dually the metric and electromagnetic 2-form define the {\it joined (total) almost-coPoisson-Jacobi phase structure}.  
 
\section{Hidden symmetries of the gravitational  contact phase structure}
\setcounter{equation}{0}
 
In \cite{JanVit12} projectable infinitesimal symmetries of the gravitational contact structure (eventually of the total almost-cosymplectic-contact structure) on the phase space were characterised. It was proved that all such symmetries are obtained as the flow 1-jet lifts of Killing vector fields (eventually the flow 1-jet lifts of Killing vector fields which are infinitesimal symmetries of the electromagnetic 2--form). In what follows we shall study   infinitesimal symmetries of the gravitational contact structure of the phase space which are not projectable. Such symmetries are not generated by infinitesimal symmetries of the spacetime and they are usually called {\it hidden symmetries}.
 
\subsection{Infinitesimal symmetries of the gravitational contact phase structure}

In what follows we assume the gravitational contact and the gravitational Jacobi structures
$(-\wha\tau,\Omega):=(-\wha\tau,\Omega[g]) $ and  $(-\wha\gamma,\Lambda):=(-\wha\gamma[g],\Lambda[g]) $
given by the metric.
Let us recall that in this situation $\Omega = - d\wha\tau$.
Then by Corollary \ref{Cr: 2.3} infinitesimal symmetries of the gravitational contact phase structure are the Hamilton-Jacobi lifts 
$
X = df\Sha + f\,\wha\gamma\,,
$
where $ f $ is a {\it conserved} phase function, i.e. $\wha\gamma.f = 0$. Moreover, $f = \wha\tau(X) = \wha\tau(\underline{X}) $. Here
$
\underline{X} = T\pi^1_0(X): \M J_1\f E \to T\f E\,
$
is a fibred morphism over $\f E$.
So any infinitesimal symmetry is the Hamilton-Jacobi lift
\begin{equation}\label{Eq: 4.1}
X = d(\wha\tau(\underline{X}))\Sha + \wha\tau(\underline{X})\,\wha\gamma\,
\end{equation} 
of the phase function $\wha\tau(\underline{X})$ where 
$
\underline{X} : \M J_1\f E \to T\f E\,
$ such that $ \wha\gamma.(\wha\tau(\underline{X}))=0 $. 

\begin{rmk}\label{Rm: 4.1}
{\rm 
Let us remark that 
$ \underline{X} = T\pi^1_0(X): \M J_1\f E \to T\f E $ is a {\it generalized (1st order) vector field} in the sense of Olver \cite{Olv86}.
}
\hfill\END
\end{rmk}

In what follows we shall characterize generalized vector fields $ \underline{X} $  admitting infinitesimal symmetries of the gravitational contact phase structure. So, such generalized vector fields have to satisfy the following two conditions:

\smallskip
1. (Projectability condition) The Hamilton-Jacobi lift 
\eqref{Eq: 4.1} of the phase function $\wha\tau(\underline{X})$  projects on  $\underline{X}$.

\smallskip
2. (Conservation condition) The phase function $\wha\tau(\underline{X})$ is conserved, i.e. $\wha\gamma.(\wha\tau(\underline{X}))=0$ .
  
\medskip
We shall start with the first projectability condition.

\begin{thm}\label{Th: 4.1}
 Let $\underline{X}:{\M J}_1\f E \to T\f E$ be a generalized vector field, then
the following assertions are equivalent:

\smallskip
1. The Hamilton-Jacobi lift $X = d(\wha\tau(\underline{X}))\Sha + \wha\tau(\underline{X})\,\wha\gamma$  projects on  $\underline{X}$.

\smallskip
2. The vertical prolongation
$$
V\underline{X}:V{\M J}_1\f E \to VT\f E = T\f E \oplus T\f E
$$
has values in the kernel of \,$\wha\tau$.

\smallskip
3. In coordinates
\begin{equation}\label{Eq: 4.2}
\br{G}^0_{0\rho}\,\der^0_j\underline{X}^\rho = 0\,.
\end{equation}
\end{thm}

\begin{proof}
1. $\Leftrightarrow$ 3.
Let 
$$
\underline{X}= \underline{X}^{\lambda}\,\der_\lambda\,,\qquad
\underline{X}^{\lambda}\in C^\infty(\M J_1\f E)\,,
$$
be a generalized vector field.
Then we have the coordinate expression
\begin{equation}\label{Eq: 4.3}
f = \wha\tau(\underline{X}) = - c_0\,\alpha^0\, \breve{G}^0_{0\rho}\,\underline{X}^\rho\,.
\end{equation}
For the function \eqref{Eq: 4.3} we have
\begin{align*}
d(\wha\tau(\underline{X}))\Sha
 & =
\big(
\underline{X}^\lambda 
+ (\alpha^0)^2\, \breve{g}_{0\rho}\,\underline{X}^\rho\,\breve{\delta}^\lambda_0
+ \breve{G}^0_{0\rho}\,\breve{G}_0^{j\lambda}\,\der^0_j\underline{X}^\rho
\big)\,\der_\lambda
\\
&\quad
-
\br G^{i\rho}_0\,\big[
\underline{X}^\sigma\,\der_\sigma\breve{G}^0_{0\rho}
+ \breve{G}^0_{0\sigma}\,\der_\rho\underline{X}^\sigma
+ (\alpha^0)^2\,\breve{g}_{0\sigma}\, \underline{X}^\sigma\,( \breve{\delta}^{\omega}_0\der_\omega\breve{G}^{0}_{0\rho} 
-  \tfrac12 \,\der_\rho\wha{G}^{0}_{00})
\\
&\quad
+ \br G^{j\omega}_0 \, 
\breve{G}^0_{0\sigma}\,  \der^0_j\underline{X}^\sigma\,(\der_\omega\breve{G}^0_{0\rho}-
\der_\rho\breve{G}^0_{0\omega}
)
 \big]\, \der^0_i
\end{align*}
and
$$
\wha\tau(\underline{X})\, \wha\gamma = - (\alpha^0)^2 \br g_{0\rho}\, \underline{X}^\rho \,\br \delta^\lambda_0\, \der_\lambda + 
(\alpha^0)^2 \br g_{0\rho}\, \underline{X}^\rho \,\br G^{i\sigma}_0\,
(\breve{\delta}^\omega_0\,\der_\omega\breve{G}^0_{0\sig} - \tfrac 12 \der_\sigma\wha G^0_{00})\,\der^0_i\,,
$$
where 
$ \wha G^{0}_{00} = G^{0}_{\lam\nu}\,\breve{\delta}^\lambda_0 \,\breve{\delta}^\mu_0  $ (the same notation we use for $ \wha{g}_{00}= g_{\lam\nu}\,\breve{\delta}^\lambda_0 \,\breve{\delta}^\mu_0   $).
So we obtain
\begin{align}\label{Eq: 4.4}
X 
&= 
d(\wha\tau(\underline{X}))\Sha +
\wha\tau(\underline{X})\, \wha\gamma
=
\big(
\underline{X}^\lambda 
+ \breve{G}^0_{0\rho}\,\breve{G}_0^{j\lambda}\,\der^0_j\underline{X}^\rho
\big)\,\der_\lambda
\\
&\quad\nonumber
-
\br G^{i\rho}_0\,\big[
\underline{X}^\sigma\,\der_\sigma\breve{G}^0_{0\rho}
+ \breve{G}^0_{0\sigma}\,\der_\rho\underline{X}^\sigma
+ \breve{G}^0_{0\sigma}\, \br G^{j\omega}_0\,\der_j^0\underline{X}^\sigma(\der_\omega\breve{G}^0_{0\rho}-
\der_\rho\breve{G}^0_{0\omega}
)
 \big]\, \der^0_i
\end{align}
and this vector field projects on $ \underline{X} $ if and only if the condition 
\begin{equation}\label{Eq: 4.5}
\breve{G}^0_{0\rho}\,\breve{G}_0^{j\lambda}\,\der^0_j\underline{X}^\rho = 0\,
\end{equation}
is satisfied.

Further let us consider the generalized vector field
$ \widetilde X:\M J_1\f E\to T\f E $ with the coordinate expression
$$
\widetilde X  =
\breve{G}^0_{0\rho}\,\breve{G}_0^{j\lambda}\,\der^0_j\underline{X}^\rho
\,\der_\lambda\,.
$$
This generalized vector field is $ \wha\tau $-vertical, i.e.
$\widetilde X: \M J_1\f E \to V_\tau\f E  $,
 and if we use the isomorphism $ G\Fla: T\f E\to \B T \ten T^{*}\f E $ we obtain
the $\B T $-valued 1-form
\begin{equation*}
\widetilde X\Fla 
=
\breve{G}^0_{0\rho}\,\breve{\delta}_\mu^{j}\,\der^0_j\underline{X}^\rho
\,u_0\ten d^{\mu}\,,
\end{equation*}
$ u_0 $ is a base of $ \B T $,  such that $ \widetilde X\Fla (\wha{\K{d}}) = 0 $, i.e.
$\widetilde X\Fla:\M J_1\f E\to \B{T}\ten {V}^*_\tau\f E $.
Further, if we use the isomorphism 
$
\nu_\tau^{-1} : \B{T}\ten {V}^*_\tau\f E\to V^*\M J_1\f E
$,
we obtain from \eqref{Eq: 3.3}
$$
\nu_\tau^{-1}(\widetilde X\Fla) 
=
c_0\,\alpha^0\,\breve{G}^0_{0\rho}\,\der^0_j\underline{X}^\rho
\,d^j_0\,.
$$
Since $\nu_\tau^{-1}\com G\Fla:{V}_\tau\f E\to V^*\M J_1\f E$
is an isomorphism we obtain
that $\nu_\tau^{-1}(\widetilde X\Fla) = 0$ if and only if $\widetilde X =0$,
i.e. \eqref{Eq: 4.2} is satisfied if and only if \eqref{Eq: 4.5}
is satisfied. 

\smallskip
2. $\Leftrightarrow$ 3.
It follows from 
$
\nu_\tau^{-1}(\widetilde X\Fla) = - \wha\tau(V\underline{X})\,.
$
\end{proof}

If the vector field \eqref{Eq: 4.1} projects on $\underline{X}$ then we get the coordinate expression
\begin{align}\label{Eq: 4.6}
X = d(\wha\tau(\underline{X}))\Sha +
\wha\tau(\underline{X})\, \wha\gamma
& =
\underline{X}^\lambda 
\,\der_\lambda
-
\br G^{i\rho}_0\,\big[
\underline{X}^\sigma\,\der_\sigma\breve{G}^0_{0\rho}
+ \breve{G}^0_{0\sigma}\,\der_\rho\underline{X}^\sigma
 \big]\, \der^0_i\,.
\end{align}

Next, let us characterize conditions for the phase function $  \wha\tau(\underline{X}) $  to be conserved.

\begin{thm}\label{Th: 4.2}
For a generalized vector field $\underline{X}$ the phase function $ \wha\tau(\underline{X}) $  is conserved if and only if the condition
\begin{equation}\label{Eq: 4.7}
0 
=
- \br\delta^\omega_0\,\br\delta^\rho_0\,\big(
G^0_{\omega\sigma}\, \der_\rho \underline{X}^\sigma
+ \tfrac 12 \underline{X}^\sigma\, \der_\sigma G^0_{\omega\rho}
\big)
+
\breve{G}^0_{0\omega}\,\breve{G}_0^{j\rho}\,\der^0_j\underline{X}^\omega\,\big(
\br\delta^\sigma_0\,\der_\sigma\breve{G}^0_{0\rho}
- \tfrac 12 \, \der_\rho\wha G^0_{00}
\big)
\end{equation}
is satisfied.
\end{thm}

\begin{proof}
Let us consider the condition  $\wha\gamma.f = 0$ which is equivalent with
$$
\tfrac{\h}{m\,c} \alpha^0\br \delta^{\rho}_{0}\,(\der_{\rho}f + \Gamma[g]_{\rho}{}^{i}_{0}\, \der^0_if) 
=
0\,.
$$
Then for the function \eqref{Eq: 4.3} we obtain
\begin{align*}
\wha\gamma.(\wha\tau(\underline{X}))
& =
- (\alpha^0)^2\br\delta^\omega_0\,\br\delta^\rho_0\,\big(
g_{\omega\sigma}\, \der_\rho \underline{X}^\sigma
+ \tfrac 12 \underline{X}^\sigma\, \der_\sigma g_{\omega\rho}
\big)
\\ &\quad 
+
(\alpha^0)^2\breve{g}_{0\omega}\,\breve{G}_0^{j\rho}\,\der^0_j\underline{X}^\omega\,\big(
\br\delta^\sigma_0\,\der_\sigma\breve{G}^0_{0\rho}
- \tfrac 12 \, \der_\rho\wha G^0_{00}
\big)
\end{align*}
which vanishes if and only if \eqref{Eq: 4.7} is satisfied.
\hfill{ }
\end{proof}

\begin{thm}\label{Th: 4.3}
Let $\underline{X}$ be
a generalized vector field 
satisfying the projectability condition. Then the following assertions are equivalent:

1. The Hamilton-Jacobi lift $X=d(\wha\tau(\underline{X}))\Sha + \wha\tau(\underline{X})\, \wha\gamma$ is an infinitesimal symmetry of the gravitational contact phase structure.

2. The phase function $ \wha\tau(\underline{X}) $ is conserved.

3. The vector field $ [\wha\gamma,X] $ is in $ \ker\wha\tau $.

4. In coordinates
 \begin{align}\label{Eq: 4.8}
0 
& = \br\delta^\omega_0\,\br\delta^\rho_0\,\big(
g_{\omega\sigma}\, \der_\rho \underline{X}^\sigma
+ \tfrac 12 \underline{X}^\sigma\, \der_\sigma g_{\omega\rho}
\big)
=
\br g_{0\sigma}\,\br\delta^\rho_0\, \der_\rho \underline{X}^\sigma
+ \tfrac 12 \underline{X}^\sigma\, \der_\sigma \wha g_{00}
\,.
\end{align}
\end{thm}

\begin{proof}
1. $ \Leftrightarrow $ 2. It follows from the fact that the Hamilton-Jacobi lift of conserved functions are infinitesimal symmetries of the contact structure.

2. $ \Leftrightarrow $ 3. For a generalized vector field $ \underline{X} $ satisfying the projectability condition we have
\begin{equation}\label{Eq: 4.9}
\wha\gamma.(\wha\tau(\underline{X})) = - (\alpha^0)^2\big(\br g_{0\sigma}\,\br\delta^\rho_0\, \der_\rho \underline{X}^\sigma
+ \tfrac 12 \underline{X}^\sigma\, \der_\sigma \wha g_{00}\big)
=  \wha\tau([\wha\gamma,X]).
\end{equation}

3. $ \Leftrightarrow $ 4.
 It follows from the coordinate expression
 \eqref{Eq: 4.9}.
\hfill{ }
\end{proof}

\begin{Example}\label{Ex: 4.1}
{\rm Let us suppose that a phase vector field $X$ is projectable on a vector field on $\f E$, i.e. $\underline{X}^\lambda\in C^\infty(\f E)$.
Then \eqref{Eq: 4.2} is satisfied identically and \eqref{Eq: 4.8} is equivalent with
$L_{\underline{X}} g = 0$. In this case the Hamilton-Jacobi lift \eqref{Eq: 4.1} coincides with the flow lift $\M J_1\underline{X}$
and projectable infinitesimal symmetries of the gravitational contact phase structure are the flow 1-jet lifts of Killing vector fields which is the result of \cite{JanVit12}.
}
\hfill\END\end{Example}

So, we can summarise.

\begin{thm}\label{Th: 4.4}
All infinitesimal symmetries of the gravitational contact phase structure are the Hamilton-Jacobi lifts of phase functions  
$ \wha\tau(\underline{X}) $, where generalized vector fields
$ \underline{X} $ satisfy the projectability and the conservation conditions.    

Moreover, if 
$ \underline{X} $
factorises through a spacetime vector field, then the corresponding infinitesimal symmetry is projectable. If $ \underline{X} $ is a generalized vector field which is not factorisable through a spacetime vector field, then the corresponding infinitesimal symmetry is hidden.
\hfill\END
\end{thm}

\begin{crl}\label{Cr: 4.5}
According to the above Theorem \ref{Th: 4.4} to give hidden symmetries of the gravitational contact phase structures we have to
find functions $ \underline{X}^{\lambda}\in C^\infty(\M J_1\f E) $ 
satisfying the conditions \eqref{Eq: 4.2} and \eqref{Eq: 4.8}.
\hfill\END
\end{crl}

\subsection{Hidden symmetries and Killing tensor fields}

In this section we give examples of nonprojectable, i.e. hidden, symmetries of the gravitational contact phase structure. Such symmetries are generated by Killing $k$-vector fields, $k>1$.

 By Corollary \ref{Cr: 2.3} constant multiples of the Reeb vector field are hidden symmetries of the gravitational contact phase structure. The Reeb vector field of the gravitational contact phase structure is $- \wha\gamma $ which implies that constant multiples of the generalized vector field 
 $\wha{\K{d}} $ satisfy the projectability and the conservation conditions. But 
 $$
\wha{\K{d}} = - \wha\tau\con \wha{\ba G} 
= c_0\, \alp^0\, \breve{G}^0_{0\rho}\,\wha{{G}}^{\rho\lambda}\,\der_\lambda
= \tfrac{\h\, \alpha^0}{m\, c}\,\breve{G}^0_{0\rho}\, {{G}}^{\rho\lambda}_0\,\der_\lambda\,,  
$$
where $\wha{\bar{G}} = \wha G^{\lambda\mu}\, \der_\lambda\ten\der_\mu$, $  \wha G^{\lambda\mu}= \tfrac{\h}{m\,c\,c_0}\, G^{\lambda\mu}_0 $, is the unscaled contravariant metric.  
 The corresponding phase function is then
$-\wha{\ba G}(\wha\tau,\wha\tau) =  1 $. Now, we generalize this situation. We consider  a symmetric $ k $-vector field $ \overset{k}{K}=\overset{k}{K}{}^{\lambda_1\dots \lam_k}\, \der_{\lambda_1}\ten \dots \ten \der_{\lambda_k} $, $k\ge 1$, on $\f E$
and the induced phase function
\begin{equation}\label{Eq: 4.10}
\overset{k}{K}(\wha\tau):=\overset{k}{K}(\wha\tau,\dots,\wha\tau)
=
(-1)^{k} (c_0\, \alpha^0)^{k}\,  \breve{G}^{0}_{0\underline{\rho}_{k}}\, \overset{k}{K}{}^{\underline{\rho}_{k}} \,,
\end{equation}
where
$\underline{\rho}_{k}= \rho_{1}\dots\rho_{k}$ is the multiindex and we denote by
$\breve{G}^0_{0\underline{\rho}_{k}}= \breve{G}^0_{0\rho_{1}}\,\dots\breve{G}^0_{0\rho_{k}}$.
Then, if we consider the Hamiltonian-Jacobi lift 
\begin{equation}\label{Eq: 4.11}
X[\overset{k}{K}] = d \big(\overset{k}{K}(\wha\tau)\big)\Sha + \overset{k}{K}(\wha\tau)\, \wha\gamma\,,
\end{equation}
we obtain the coordinate expression of its projection $\underline{X}[\overset{k}{K}] = T\pi^1_0(X[\overset{k}{K}])$ in the form
\begin{align*}
\underline{X}[\overset{k}{K}]
& = 
(-1)^{k-1} (c_0\, \alpha^0)^{k-1}\, \big[k\,\br G^0_{0\underline{\rho}_{k-1}}\, \overset{k}{K}{}^{\underline{\rho}_{k-1}\lambda}
+ {(k-1)} \tfrac{\h_0}{m} (\alpha^0)^2 \br G^{0}_{0\underline{\rho}_{k}}\, \overset{k}{K}{}^{\underline{\rho}_{k}}\,\br\delta^\lambda_0
\big]\,\der_\lambda
\,
\end{align*}
which is the coordinate expression of the  generalized vector field
\begin{equation}\label{Eq: 4.13}
\underline{X}[\overset{k}{K}] = k\, \underset{(k-1)-times}{\underbrace{\wha\tau\con \dots\con \wha\tau \con}}\overset{k}{K} - {(k-1)}\, \overset{k}{K}(\wha\tau)\, \wha{\K{d}}: \M J_1\f E\to T\f E\,.
\end{equation}
It is easy to see that \eqref{Eq: 4.13} satisfies the projectability condition since $\wha\tau(\underline{X}[\overset{k}{K}]) = \overset{k}{K}(\wha\tau)$.

Now, it is sufficient to find conditions for phase functions $\overset{k}{K}(\wha\tau)$ to be conserved, i.e. $\wha\gamma.\overset{k}{K}(\wha\tau) = 0$.

\begin{lem}\label{Lm: 4.6}
The equation $\overset{k}{K}(\wha\tau) =0$, $k\ge 1$, is satisfied if and only if $ \overset{k}{K} \equiv 0 $. 
\end{lem}

\begin{proof}
We have 
$$
\overset{k}{K}(\wha\tau) = (-c_0\,\alpha^{0})^{k} \, \overset{k}{K}{}^{\underline{\lambda}_k}\,
\breve{G}^{0}_{0\underline{\lam}_k}
= \big(-\tfrac{m\,c_0\,\alpha^{0}}{\h_0}\big)^{k} \, \overset{k}{K}{}^{\lambda_1\dots\lambda_k}\,
{g}_{\rho_1\lam_1}\dots {g}_{\rho_k\lam_k}\,\breve{\delta}^{\rho_1}_0\dots\breve{\delta}^{\rho_k}_0
\,.
$$
Then $\overset{k}{K}(\wha\tau) =0$ if and only if
$$
\overset{k}{K}{}^{\lambda_1\dots\lambda_k}\,
{g}_{\rho_1\lam_1}\dots {g}_{\rho_k\lam_k}\,\breve{\delta}^{\rho_1}_0\dots\breve{\delta}^{\rho_k}_0 = 0\,.
$$ 
The above function is polynomial on fibres of the phase space with coefficients to be
spacetime functions so it vanishes if and only if $\overset{k}{K}{}^{\lambda_1\dots\lambda_k}\,
{g}_{\rho_1\lam_1}\dots {g}_{\rho_k\lam_k}=0$ for all indices $ \rho_1 , \dots ,\rho_k $ and since the metric is regular it is satisfied if and only if $\overset{k}{K}{}^{\lambda_1\dots\lambda_k}=0$ for all indices $ \lambda_1 , \dots ,\lambda_k $.
\end{proof}

\begin{thm}
The phase function $ \overset{k}{K}(\wha\tau) $, $k\ge 1$, is conserved with respect to the gravitational Reeb vector field, i.e. $ \wha\gamma.\overset{k}{K}(\wha\tau) = 0 $, if and only if $\overset{k}{K}$ is a Killing $k$--vector field.
\end{thm}

\begin{proof}
We have 
$$
\wha\gamma.\overset{k}{K}(\wha\tau) =(\wha\gamma.\overset{k}{K}{}^{\lambda_1\dots\lambda_k})\,\wha\tau_{\lam_1}\dots \wha\tau_{\lam_k}
+ k\, \overset{k}{K}{}^{\rho\lambda_1\dots\lambda_{k-1}}\,(\wha\gamma.\wha\tau_\rho)\, \wha\tau_{\lam_1}\dots \wha\tau_{\lam_{k-1}}\,.
$$
But we have 
$$
\wha\gamma.\wha\tau_\rho = \gamma.\tau_\rho= - \tfrac 12 (\alpha^{0})^{2}\, \der_\rho\wha g_{00}\,,
$$
which implies, by using the identities $ \breve{\delta}^\rho_0 = g^{\rho\sigma}\breve{g}_{0\sigma} $ and $ \der_\rho\wha g_{00} = - \breve{g}_{0\sigma_1}\,\breve{g}_{0\sigma_2}\,\der_\rho g^{\sigma_1\sigma_2} $,
\begin{align}\label{Eq: a}
\wha\gamma.\overset{k}{K}(\wha\tau) 
& =
- \tfrac{\h^{2}}{m^{2}\, c^{2}}\,\big(
g^{\rho\lambda_1}\, \der_\rho\overset{k}{K}{}^{\lambda_2\dots\lambda_{k+1}}
- \tfrac k2 \, \overset{k}{K}{}^{\rho\lambda_1\dots\lambda_{k-1}}\, \der_\rho g^{\lam_k\lambda_{k+1}}
\big)\,\wha\tau_{\lam_1}\dots \wha\tau_{\lam_{k+1}}
\\
& = \nonumber
- \tfrac{1}{2}\,[\wha{\bar G},\overset{k}{K}](\wha\tau)
= \tfrac{1}{2}\,[\overset{k}{K},\wha{\bar G}](\wha\tau)
\,.
\end{align}
So, by Lemma \ref{Lm: 4.6}, the function $ \overset{k}{K}(\wha\tau) $ is conserved if and only if $ \overset{k}{K} $ is a Killing $k$--vector field.
\end{proof}
 
\begin{crl}\label{Th: 4.6}
Let $ \overset{k}{K}$ be a Killing $ k $-vector field, $ k\ge 1 $, and $ \underline{X}[ \overset{k}{K}] $ is
the induced generalized vector field \eqref{Eq: 4.13}. Then the Hamilton-Jacobi lift  $X[ \overset{k}{K}]$ given by \eqref{Eq: 4.11} is an infinitesimal symmetry of the gravitational contact phase structure. Moreover, for $k=1$ this infinitesimal symmetry is projectable on the Killing vector field $ \overset{1}{K}$.
For $ k\ge 2 $ the corresponding infinitesimal symmetry is hidden and, especially,
for $ \overset{2}{K}=k\,\wha{\bar G}\,,$ $k\in\mathbb{R}\,,$ the corresponding hidden symmetry is $ - k\,\wha\gamma  $, i.e. a constant multiple of the Reeb vector field $- \wha\gamma$.
\end{crl}

\begin{proof}
It is easy to see that for $k=1$ we obtain $\underline{X}[ \overset{1}{K}] =  \overset{1}{K}$ which is a spacetime Killing vector field.
\end{proof}

\begin{rmk}
{\rm 
The pullback of a spacetime function $  \overset{0}{K} $ is conserved if and only if it is a constant. Then for a constant $\overset{0}{K} $ and Killing $k$-vector fields $\overset{k}{K}$, $k\ge 1$, the conserved phase function of the type $ K = \overset{0}{K} + \sum_{k\ge 1} \overset{k}{K}(\wha\tau) $ admits the hidden infinitesimal symmetry
$$
X[K] = \overset{0}{K}\, \wha \gamma + \sum_{k\ge 1} {X}[\overset{k}{K}]
$$
which projects on the generalized vector field
$$
\underline{X}[K] = \big(\overset{0}{K}-\sum_{k\ge 2}(k-1)\,\overset{k}{K}(\wha\tau) \big) \, \wha{\K{d}} + \overset{1}{K} + \sum_{k\ge 2} k\, \underset{(k-1)-times}{\underbrace{\wha\tau\con \dots\con \wha\tau \con}}\overset{k}{K}\,
$$
satisfying the projectability and the conservation conditions.
}
\hfill\END
\end{rmk}

\section{Lie algebra of infinitesimal symmetries of the gravitational contact phase structure}
\setcounter{equation}{0}
Phase infinitesimal symmetries form a Lie algebra with respect to the Lie bracket of vector fields. It follows from
$$
L_{[X,Y]} = L_X\, L_Y - L_Y\, L_X\,.
$$
\subsection{General situation}
 
In this section we assume a generalized vector field  $\underline{X}$
and by $X$ we denote the vector field \eqref{Eq: 4.1}.
So if $\underline{X}$ and $\underline{Y}$ are generalized vector fields satisfying the projectability and the conservation conditions \eqref{Eq: 4.2} and 
\eqref{Eq: 4.8}, then the generalized vector field $\underline{[X,Y]}$
has to satisfy the projectability and the conservation conditions too. The projectability condition can be checked directly in coordinates and we shall omit it.

\begin{lem}\label{Lm: 5.1}
For generalized vector fields $ \underline{X} $
and $ \underline{Y} $ satisfying the projectability condition  we
have
\begin{equation}\label{Eq: 5.1}
\{\wha\tau(\underline{X}),\wha\tau(\underline{Y})\}+\wha\tau(\underline{X})\, \wha\gamma.(\wha\tau(\underline{Y})) -
\wha\tau(\underline{Y})\, \wha\gamma.(\wha\tau(\underline{X}))
= 
\wha\tau(\underline{[X,Y]}) 
\,.
\end{equation}
\end{lem}

\begin{proof}
For any phase functions $ f,g $ we have
$$
\{f,g\}
=
\tfrac{1}{c_0\, \alpha^{0}}\big[
\breve{G}^{j\lambda}_{0}\, (\der_\lambda f\, \der^{0}_j g - 
\der_\lambda g\, \der^{0}_j f)
- \breve{G}^{i\rho}_{0}\,\breve{G}^{j\sigma}_{0}\,
(\der_\sigma \breve{G}_{0\rho}^{0}- \der_\rho \breve{G}_{0\sigma}^{0})\, \der^{0}_i f\, \der^{0}_j g
\big]\,.
$$
Then for the phase functions $ \wha\tau(\underline{X}) $ and $ \wha\tau(\underline{Y}) $ where
$ \underline{X} $
and $ \underline{Y} $ satisfy the projectability condition  we get
\begin{align*}
\{\wha\tau(\underline{X}),  \wha\tau(\underline{X})\}
& =  
{c_0\, \alpha^{0}}
\breve{G}_{0\rho}^{0}\, (\underline{Y}^\sigma \, \der_\sigma\underline{X}^\rho
-
\underline{X}^\sigma \, \der_\sigma\underline{Y}^\rho)
\\
&\quad 
+ c_0\,(\alpha)^{3}\,\big[\breve{g}_{0\sigma}\,\underline{Y}^\sigma
\big(\breve{G}_{0\rho}^{0}\,\breve{\delta}_{0}^{\lam}\,
\der_\lambda   \underline{X}^\rho + 
\tfrac 12 \underline{X}^\rho\,\der_\rho\wha G^{0}_{00}
\big)
\\
&\qquad
 -
\breve{g}_{0\rho}\,\underline{X}^\rho
\big(\breve{G}_{0\sigma}^{0}\,\breve{\delta}_{0}^{\lam}\,
\der_\lambda   \underline{Y}^\sigma + 
\tfrac 12 \underline{Y}^\sigma\,\der_\sigma\wha G^{0}_{00}
\big)
\big]\,.
\end{align*}
Further
\begin{align*}
 \wha\tau(\underline{X})\, \wha\gamma.\wha\tau(\underline{Y})
& -
\wha\tau(\underline{Y})\, \wha\gamma.\wha\tau(\underline{X})
= 
\\
& =
c_0\,(\alpha^{0})^{3}\big[\breve{g}_{0\rho}\,\underline{X}^{\rho}\, \big(
\breve{G}_{0\sigma}^{0}\,\breve{\delta}_{0}^{\lam}\,
\der_\lambda   \underline{Y}^\sigma + 
\tfrac 12 \underline{Y}^\sigma\,\der_\sigma\wha G^{0}_{00}
\big)
\\
&\quad -
\breve{g}_{0\sigma}\,\underline{Y}^{\sigma}\, \big(
\breve{G}_{0\rho}^{0}\,\breve{\delta}_{0}^{\lam}\,
\der_\lambda   \underline{X}^\rho + 
\tfrac 12 \underline{X}^\rho\,\der_\rho\wha G^{0}_{00}
\big)
\big]\,.
\end{align*}

On the other hand 
\begin{align*}
\underline{[X,Y]}
& =
\big[
\underline{X}^{\rho}\,\der_{\rho}\underline{Y}^{\lambda}
-  \br G^{j\rho}_0\,(
\underline{X}^\sigma\,\der_\sigma\breve{G}^0_{0\rho}
+ \breve{G}^0_{0\sigma}\,\der_\rho\underline{X}^\sigma
 )\,   \der^0_j\underline{Y}^{\lambda}
 \\
& \quad
 - \underline{Y}^{\rho}\,\der_{\rho}\underline{X}^{\lambda}
+ \br G^{j\rho}_0\,(
\underline{Y}^\sigma\,\der_\sigma\breve{G}^0_{0\rho}
+ \breve{G}^0_{0\sigma}\,\der_\rho\underline{Y}^\sigma
 )\,   \der^0_j\underline{X}^{\lambda}  
\big]\, \der_\lambda\,.
\end{align*}
and
\begin{align*}
\wha\tau(\underline{[X,Y]}) 
& =
 {c_0\, \alpha^{0}}
\breve{G}_{0\rho}^{0}\, (\underline{Y}^\sigma \, \der_\sigma\underline{X}^\rho
-
\underline{X}^\sigma \, \der_\sigma\underline{Y}^\rho)\,
\end{align*}
which proves Lemma \ref{Lm: 5.1}.
\hfill{ }
\end{proof}

\begin{crl}\label{Cr: 5.2}
For generalized vector fields $ \underline{X} $
and $ \underline{Y} $ satisfying the projectability and the conservation conditions we
have
\begin{equation}\label{Eq: 5.2}
\{\wha\tau(\underline{X}),\wha\tau(\underline{Y})\}
= 
\wha\tau(\underline{[X,Y]}) 
\,.
\end{equation}
Moreover, the phase function $ \wha\tau(\underline{[X,Y]})  $ is conserved.
\end{crl}

\begin{proof}
\eqref{Eq: 5.2} follows from Lemma \ref{Lm: 5.1} and the fact  that $\wha\tau(\underline{X})$ and $\wha\tau(\underline{Y})$ are conserved. Further $ \wha\tau(\underline{[X,Y]})  $ is conserved because it is the Poisson bracket of two conserved phase functions which is by \eqref{Eq: 2.7} conserved.
\end{proof}

\begin{rmk}\label{Rm: 5.1}
{\rm 
If we have a phase vector field $X$ then it defines the unique
generalized vector field   $\underline{X}= \pi^{1}_{0}(X) $.
On the other hand given a generalized vector fields $\underline{X}$
there are many phase vector fields projectable on $\underline{X}$.
So there is no uniquely defined generalized vector field
$\underline{[X,Y]} $, where $X$ and $Y$ are some phase vector fields
projectable on $\underline{X}$ and $\underline{Y}$, respectively.
 The projectability condition  implies that one of the phase vector fields projectable on $\underline{X}$
 is the Hamilton-Jacobi lift of $ \wha\tau(\underline{X}) $, so for generalized vector fields satisfying  the projectability condition
we have distinguished phase vector fields which allow us to define 
the Lie bracket for generalized vector fields satisfying  the projectability condition by the condition $[\underline{X}, \underline{Y}] = \underline{[X,Y]} $, where as $X$ and $Y$ we assume the Hamilton-Jacobi lifts \eqref{Eq: 4.1}.
}
\hfill\END
\end{rmk}

\begin{rmk}\label{Rm: 5.2}
{\rm 
Lemma \ref{Lm: 5.1} means that the rule transforming generalized vector fields $\underline{X}$ satisfying the projectability condition into the 
 phase functions $\wha\tau(\underline{X})$ is a Lie algebra morphism with respect to the Lie bracket of generalized vector fields and the Jacobi bracket of phase functions.

On the other hand Corollary \ref{Cr: 5.2} means that the rule transforming generalized vector fields $\underline{X}$ 
satisfying both projectability and conservation conditions into
 the phase functions $\wha\tau(\underline{X})$ is a Lie algebra morphism with respect to the Lie bracket of generalized vector fields and the Poisson bracket of conserved phase functions.
 }
 \hfill\END
\end{rmk}

\begin{thm}\label{Th: 5.3}
For generalized vector fields $ \underline{X} $
and $ \underline{Y} $ such that $X$ and $Y$ are hidden symmetries
of the gravitational contact phase structure the vector field $[X,Y]$  is given as
\begin{equation}\label{Eq: 5.3}
 [X,Y]
 =
 d\wha\tau(\underline{[X,Y]}) \Sha + \wha\tau(\underline{[X,Y]})\, \wha\gamma
\end{equation}
and the phase function $\wha\tau(\underline{[X,Y]})$ is conserved. 
\end{thm}

\begin{proof}
From Lemma \ref{Lm: 2.5} we get
$$
[X,Y] = d\{\wha\tau(X),\wha\tau(Y) \}\Sha + \{\wha\tau(X),\wha\tau(Y) \}\, \wha\gamma
$$
and by Corollary \ref{Cr: 5.2} the Lie bracket $[X,Y]$ is given by
\eqref{Eq: 5.3} where the phase function $\wha\tau(\underline{[X,Y]})$  is conserved.
\end{proof}

\begin{rmk}\label{Rm: 5.3}
{\rm
 The above Theorem \ref{Th: 5.3} then means that 
the rule transforming generalized vector fields satisfying the projectability and the conservation  conditions into hidden symmetries of the gravitational contact phase structure is a Lie algebra homomorphism with respect to the Lie bracket of generalized vector fields and the Lie bracket of hidden symmetries.
}
\hfill\END
\end{rmk} 
 
\subsection{Brackets of hidden symmetries generating by Killing multivectors}

Now, let us consider generalized vector fields generated by symmetric and Killing multivector fields, i.e. for a symmetric $k$-vector field $\overset{k}{K}$ we assume the generalized vector field $ \underline{X}[\overset{k}{K}] $ given by \eqref{Eq: 4.13} and the corresponding phase function $\overset{k}{K}(\wha\tau)$.

\begin{lem}\label{Lm: 5.4}
For a symmetric $k$-vector field $\overset{k}{K}$ and a symmetric $l$-vector field $\overset{l}{L}$ we have
\begin{align}
\wha\tau(\underline{[X[\overset{k}{K}],X[\overset{l}{L}]]})
& =
[\overset{k}{K},\overset{l}{L}](\wha\tau) 
+ \tfrac{l-1}{2} \overset{l}{L}(\wha\tau)\, [\overset{k}{K},\wha{\bar G}](\wha\tau)
- \tfrac{k-1}{2} \overset{k}{K}(\wha\tau)\, [\overset{l}{L},\wha{\bar G}](\wha\tau)
\\
& = \nonumber
[\overset{k}{K},\overset{l}{L}](\wha\tau) 
+ {(l-1)} \overset{l}{L}(\wha\tau)\, \wha\gamma.\overset{k}{K}(\wha\tau)
- {(k-1)} \overset{k}{K}(\wha\tau)\, \wha\gamma. \overset{l}{L}(\wha\tau)
\,.
\end{align}
\end{lem}

\begin{proof}
It can be proved in coordinates.
\end{proof}

\begin{crl}\label{Cr: 5.5}
For a symmetric $k$-vector field $\overset{k}{K}$ and a symmetric $l$-vector field $\overset{l}{L}$ we have
$$
\{\overset{k}{K}(\wha\tau),\overset{l}{L}(\wha\tau)\}
 + k\, \overset{k}{K}(\wha\tau)\, \wha\gamma.\overset{l}{L}(\wha\tau) - l\, \overset{l}{L}(\wha\tau)\, \wha\gamma.\overset{k}{K}(\wha\tau)
=[\overset{k}{K},\overset{l}{L}](\wha\tau)
$$
\end{crl}

\begin{proof}
It follows from Lemma \ref{Lm: 5.1} and Lemma \ref{Lm: 5.4}.
\end{proof}

\begin{rmk}
{\rm
Let us consider phase functions generated by symmetric multivector fields. Then we can define the bracket
\begin{equation}\label{Eq: 5.5}
[\overset{k}{K}(\wha\tau),\overset{l}{L}(\wha\tau)]_{k,l}
=
\{\overset{k}{K}(\wha\tau),\overset{l}{L}(\wha\tau)\}
 + k\, \overset{k}{K}(\wha\tau)\, \wha\gamma.\overset{l}{L}(\wha\tau) - l\, \overset{l}{L}(\wha\tau)\, \wha\gamma.\overset{k}{K}(\wha\tau)\,.
\end{equation}
By Corollary \ref{Cr: 5.5} it follows that the sheaf of such phase function is closed with respect to the above bracket. Moreover, this bracket is a Lie bracket, really it is antisymmetric and we have
$$
\sum_{cyclic} \big[[\overset{k}{K}(\wha\tau),\overset{l}{L}(\wha\tau) ]_{k,l}, \overset{m}{M}(\wha\tau)\big]_{k+l-1,m} =\sum_{cyclic} \big[[\overset{k}{K}(\wha\tau),\overset{l}{L}(\wha\tau) ], \overset{m}{M}(\wha\tau)\big] = 0\,,
$$
i.e. the Jacobi identity is satisfied. Here the bracket $[\overset{k}{K}(\wha\tau),\overset{l}{L}(\wha\tau) ] $ is the standard Jacobi bracket given by the gravitational Jacobi pair.
Then Corollary \ref{Cr: 5.5} means that the rule transforming symmetric multivector fields into phase functions is a Lie algebra homomorphism with respect to the Schouten bracket of symmetric multivector fields and the bracket given by \eqref{Eq: 5.5}. 
}
\hfill\END
\end{rmk}

In what follows we shall assume Killing multivector fields and the corresponding phase functions.

\begin{lem}\label{Lm: 5.5}
For a Killing $k$-vector field $\overset{k}{K}$ and a Killing $l$-vector field $\overset{l}{L}$ we have
\begin{align*}
\wha\tau(\underline{[X[\overset{k}{K}],X[\overset{l}{L}]]})
& =
[\overset{k}{K},\overset{l}{L}](\wha\tau)
\,.
\end{align*}
\end{lem}

\begin{proof}
It follows from Lemma  \ref{Lm: 5.4} and the fact that for Killing multivector fields the functions $\overset{k}{K}(\wha\tau) $ and  $\overset{l}{L}(\wha\tau) $ are conserved. 
\end{proof}

\begin{rmk}\label{Rm: 5.4}
{\rm
From Corollary \ref{Cr: 5.5} we obtain for Killing multivector fields
$$
\{\overset{k}{K}(\wha\tau),\overset{l}{L}(\wha\tau)\} 
= [\overset{k}{K},\overset{l}{L}](\wha\tau)
$$
and the rule transforming a Killing multivector field $\overset{k}{K}$ into the conserved phase function
$\overset{k}{K}(\wha\tau)$ is a Lie algebra homomorphism with respect to the Schouten bracket of Killing multivector fields and the Poisson bracket of conserved phase functions.
}
\hfill\END
\end{rmk}

\begin{thm}\label{Th: 5.5}
Let $ \overset{k}{K} $ be a Killing $ k $-vector field and $ \overset{l}{L} $ be a Killing $ l $-vector field, then
\begin{equation}\label{Eq: 5.4}
\big[X[\overset{k}{K}],X[\overset{l}{L}]\big] = X\big[[\overset{k}{K},\overset{l}{L}]\big]\,.
\end{equation}
i.e. the rule transforming Killing multivector fields into hidden symmetries of the gravitational contact phase structure is a Lie algebra homomorphism with respect to the Schouten bracket of Killing multivector fields and the Lie bracket of phase vector fields.
\end{thm}

\begin{proof}
From Lemma \ref{Lm: 2.5} and Remark \ref{Rm: 5.4} we have
\begin{align*}
\big[X[\overset{k}{K}],X[\overset{l}{L}]\big] 
& = d\big(\{\overset{k}{K}(\wha\tau),\overset{l}{L}(\wha\tau) \}\big)\Sha +  \{\overset{k}{K}(\wha\tau),\overset{l}{L}(\wha\tau) \}\, \wha\gamma
\\
& =
d\big([\overset{k}{K},\overset{l}{L}](\wha\tau)\big)\Sha +  [\overset{k}{K},\overset{l}{L}](\wha\tau)\, \wha\gamma
 = 
X\big[[\overset{k}{K},\overset{l}{L}]\big]\,.
\end{align*}
\vskip-1.5\baselineskip
\end{proof}

\end{document}